\documentclass[11pt]{amsart}
\usepackage{amsmath}
\usepackage{graphicx}
\usepackage{esint}

\theoremstyle{plain}
\newtheorem{theorem}{Theorem}[section]
\newtheorem{lemma}[theorem]{Lemma}

\newtheorem{cor}[theorem]{Corollary}

\theoremstyle{definition}

\numberwithin{equation}{section}

\newcommand{\tr}{\operatorname{Tr}}

\begin{document}

\title[A Penrose-Type Inequality] {A Penrose-Type Inequality with Angular Momentum and Charge for Axisymmetric Initial Data}

\author[Khuri]{Marcus Khuri}
\author[Sokolowsky]{Benjamin Sokolowsky}
\address{Department of Mathematics\\
Stony Brook University\\
Stony Brook, NY 11794, USA}
\email{khuri@math.sunysb.edu, bsokolowsky@math.sunysb.edu}

\author[Weinstein]{Gilbert Weinstein}
\address{Department of Physics and Department of Mathematics\\
Ariel University of Samaria\\
Ariel, 40700, Israel}
\email{gilbertw@ariel.ac.il}

\thanks{M. Khuri acknowledges the support of NSF Grant DMS-1708798.}

\begin{abstract}
A lower bound for the ADM mass is established in terms of angular momentum, charge, and horizon area in the context of maximal, axisymmetric initial data for the Einstein-Maxwell equations which satisfy the weak energy condition. If, on the horizon, the given data agree to a certain extent with the associated model Kerr-Newman data, then the inequality reduces to the conjectured Penrose inequality with angular momentum and charge. In addition, a rigidity statement is also proven whereby equality is achieved if and only if the data set arises from the canonical slice of a Kerr-Newman spacetime.
\end{abstract}
\maketitle

\section{Introduction}
\label{sec1}\setcounter{equation}{0}
\setcounter{section}{1}

Consider a simply connected, asymptotically flat initial data set $(M,g,k,E,B)$ for the Einstein-Maxwell equations. Here $M$ is a Riemannian 3-manifold with metric $g$, $k$ is a symmetric 2-tensor representing the second fundamental form of the embedding into spacetime, and $(E,B)$ represents the electromagnetic field. The non-electromagnetic matter energy and momentum densities are given by
\begin{equation}\label{1}
16\pi\mu_{em} = R+(\tr_{g}k)^{2}-|k|_{g}^{2}-2(|E|_g^2+|B|_g^2),\qquad
8\pi J_{em} = \operatorname{div}_{g}(k-(\tr_{g}k)g)+2E\times B,
\end{equation}
where $R$ is the scalar curvature and $E\times B$ represents cross product. It will be assumed that the weak energy condition $\mu_{em}\geq 0$ holds, the data are maximal $\tr_{g}k=0$, and that there is no charged matter
\begin{equation}\label{1.1}
\operatorname{div}_{g}E=\operatorname{div}_{g}B=0.
\end{equation}
In addition, the data are taken to be axisymmetric in that the isometry group of $(M,g)$ admits a subgroup isomorphic to $U(1)$, such that all other quantities defining the data are invariant under this $U(1)$ action. The Killing field generator will be denoted by $\eta$. Moreover, we will say that the initial data are asymptotically flat if there exists an end $M_{\text{end}}\subset M$ diffeomorphic to $\mathbb{R}^{3}\setminus\text{Ball}$, so that for some $\epsilon>0$ in the asymptotic coordinates
\begin{equation}\label{ilil}
g_{ij}=\delta_{ij}+O_{\ell}(r^{-\frac{1}{2}-\epsilon}),\text{ }\text{ }\text{ }\text{ }\partial g_{ij}\in
L^{2}(M_{\text{end}}),\text{
}\text{ }\text{ }
\text{ }k_{ij}=O_{\ell-1}(r^{-\lambda-\frac{1}{2}}),
\end{equation}
\begin{equation}
\mu_{em}, J_{em}^{i},J_{em}(\eta)\in L^{1}(M_{\text{end}}),\text{ }\text{ }\text{ }\text{ }\text{ }
E^{i}=O_{\ell-1}(r^{-\lambda}),\text{ }\text{ }\text{ }\text{ }\text{ }B^{i}=O_{\ell-1}(r^{-\lambda}),\text{
}\text{ }\text{ }
\text{ }\lambda>\frac{3}{2},\quad \ell\geq 5.
\end{equation}

Heuristic arguments originally due to Penrose \cite{Penrose} suggest the inequality
\begin{equation}\label{penrose}
m \geq \sqrt{\frac{\mathcal{A}}{16\pi}+\frac{Q^2}{2}
+\frac{\pi(Q^4+4\mathcal{J}^2)}{\mathcal{A}}} \quad\quad\text{ whenever }\quad\mathcal{A}\geq 4\pi\sqrt{Q^4+4\mathcal{J}^2},
\end{equation}
where $m$ is the ADM mass, $\mathcal{A}$ is the event horizon cross-sectional area, and the total angular momentum and charges take the form
\begin{equation}
\mathcal{J}=\frac{1}{8\pi}
\int_{S_\infty}(k_{ij}-(\tr_{g}k)g_{ij})\nu^i\eta^j,\quad\quad
Q_e=\frac{1}{4\pi}
\int_{S_\infty}E_i\nu^i,\quad\quad Q_b=\frac{1}{4\pi}
\int_{S_\infty}B_i\nu^i,
\end{equation}
with $Q^2=Q_e^2+Q_b^2$. In these formulas $S_\infty$ represents the limit as $r\rightarrow \infty$ for coordinate spheres $S_r$ in the asymptotic end, and $\nu$ is the unit outer normal. Inequality \eqref{penrose} was proposed as a check on the final state conjecture and weak cosmic censorship, in that a counterexample would essentially disprove at least one of these grand conjectures. Details concerning the heuristic derivation of this most general form of the Penrose inequality are provided in \cite{DainKhuriWeinsteinYamada}.
Furthermore an independent heuristic motivation for this inequality, based on Bekenstein's entropy bound \cite{Bekenstein}, has been given in \cite{JaraczKhuri}.

In order to prove Penrose type inequalities it is customary to replace $\mathcal{A}$ in the maximal case with the area of the outermost minimal surface, and in the general case with the minimum area required to enclose the outermost apparent horizon. Therefore, the manifold $(M,g)$ will be taken to have a boundary consisting of a single component minimal surface. Note that simple connectivity then implies that the boundary must be topologically a 2-sphere, regardless of whether this surface is stable. Moreover, the auxiliary inequality of \eqref{penrose} is not needed in the single black hole case, since
when the minimal surface is stable the area-angular momentum-charge inequality is known to be automatically satisfied \cite{Dain,DainGabach,ClementJaramilloReiris}.

The Penrose inequality without angular momentum and charge was established in the time-symmetric case through the ground breaking work of Bray \cite{Bray} and
Huisken-Ilmanen \cite{HuiskenIlmanen}. As shown in \cite{WeinsteinYamada}, the addition of charge to this inequality requires the additional assumption of the
area-charge inequality in the multiple black hole case. This version of the Penrose inequality was then established
in \cite{KhuriWeinsteinYamada,KhuriWeinsteinYamada1} by generalizing Bray's conformal flow. Inequalities providing a lower bound for the mass in terms of
angular momentum and charge, which are implied by \eqref{penrose}, have also been thoroughly established \cite{ChruscielCosta,ChruscielLiWeinstein,Costa,Dain0,KWcharge,SchoenZhou}. However, it turns out that including horizon area together with angular momentum is quite difficult. In fact, there appear to be only two results in the literature to date in this direction \cite{Anglada,Anglada1}, and the approach taken in those articles is based on inverse mean curvature flow. In contrast, the present paper focuses on the techniques used to establish the mass-angular momentum inequalities, namely minimizing renormalized harmonic energies. We refer the reader to the excellent survey \cite{Mars} for a more detailed account concerning the status of the Penrose inequality.

The results presented here rely on the existence of Weyl coordinates. These are
cylindrical type coordinates $(\rho,z,\phi)$ with $\rho\geq 0$, $-\infty<z<\infty$, $0\leq\phi<2\pi$ that are typically associated with the study of stationary axisymmetric black holes, and play an important role in that setting by helping to reduce the Einstein equations to the study of a harmonic map. Details describing this coordinate system for the present situation are given in the appendix, and are discussed in the next section. It has been shown in \cite{ChruscielNguyen} that such a coordinate system exists more generally for simply connected, asymptotically flat initial data sets. In these coordinates the metric takes the form
\begin{equation}\label{brill}
g=e^{-2U+2\alpha}(d\rho^{2}+dz^{2})+\rho^{2}e^{-2U}(d\phi+A_\rho d\rho+A_z dz)^{2},
\end{equation}
where $\eta=\partial_\phi$ is the rotational Killing field, and all the coefficient functions are smooth and axisymmetric.
In these coordinates the minimal surface horizon is identified with the interval  $(-m_0,m_0)$ on the $z$-axis. The constant $m_0>0$ is uniquely determined by
the geometry of the initial data, and $2m_0$ will be referred to as the horizon rod length. The functions $U$ and $\alpha$ exhibit singular behavior at the horizon, and this may be modeled by the corresponding functions $U_0$, $\alpha_0$ arising from the Schwarzschild solution having mass $m_0$. We may then write $U=U_0+\overline{U}$ and $\alpha=\alpha_0+\overline{\alpha}$, where the remainders $\overline{U}$ and $\overline{\alpha}$ are now uniformly bounded and possess bounded first derivatives even at the horizon. These `renormalized' functions measure the deviation from the Schwarzschild solution. An important combination of these two which appears in the horizon area formula is $\overline{\beta}:=\overline{\alpha}-2\overline{U}$. For the Kerr black hole this quantity may be expressed nicely in terms of surface gravity, see Appendix \ref{sec8}. In what follows the ADM mass/energy of the initial data will be denoted by $m$, and we note that the asymptotically flat asymptotics \eqref{ilil} guarantee that total mass and energy agree, that is the ADM linear momentum vanishes. Our main result may then be stated as follows.

\begin{theorem}\label{maintheorem}
Let $(M,g,k,E,B)$ be a simply connected, axisymmetric, maximal, asymptotically flat initial data set for the Einstein-Maxwell equations with minimal surface boundary, having nonnegative energy density $\mu_{em}\geq 0$, no charged matter, and satisfying the compatibility condition for the existence of a twist potential $J_{em}(\eta)=0$. Let $A_k$ and $\overline{\beta}_k$ denote the horizon area and Weyl coordinate function for the unique Kerr-Newman black hole sharing the same angular momentum, charge, and horizon rod length as the initial data set. Then
\begin{equation}\label{4123}
m\geq   \sqrt{\frac{A_{k}}{16\pi}+\frac{Q^2}{2}+\frac{\pi(Q^4+4\mathcal{J}^2)}{A_{k}}} +\frac{1}{4}\int_{-m_0}^{m_0}(\overline{\beta}(0,z)-\overline{\beta}_{k}(0,z))dz,
\end{equation}
and equality occurs if and only if the initial data agree with that of the corresponding Kerr-Newman spacetime.
\end{theorem}

The hypotheses of this theorem are in agreement with those expected for the conjectured Penrose inequality with angular momentum and charge, except for one missing statement. Namely, in the above result the minimal surface boundary is not required to be outerminimizing, meaning it is not required to have the property that every surface which encloses it has area greater than or equal to $A=|\partial M|$. This property is necessary, however, for the actual Penrose inequality as counterexamples are known to exist without it. Thus, Theorem \ref{maintheorem} holds under more general circumstances than those for which the Penrose inequality can be valid, and so the resulting inequality \eqref{4123} must differ from \eqref{penrose}. Indeed, the most apparent difference arises from the presence of the horizon rod integral involving the functions $\overline{\beta}$ and $\overline{\beta}_k$, which does not appear in the Penrose inequality. This integral measures the discrepancy between the initial data and the model Kerr-Newman solution. It is unknown at this time whether this horizon integral is nonnegative under the current hypotheses. One may speculate that nonnegativity is not necessarily guaranteed unless the boundary is outerminimizing. After all $\overline{\beta}$, like the outerminimizing condition, is non-local.
Another difference between \eqref{4123} and the conjectured inequality is the presence of the Kerr-Newman horizon area $A_k$ instead of $A$, although the algebraic structure of this part of the inequality is the same. Despite these differences, one may achieve the desired Penrose inequality under additional
assumptions. In particular, if we assume that the initial data is appropriately
similar to the model Kerr-Newman solution at the horizon then the conjectured inequality follows.

\begin{cor}\label{corollary}
Under the hypotheses of Theorem \ref{maintheorem}, assume further that
$A\geq A_{k}$ and  $\overline{\beta}$ is constant on the horizon rod, then
\begin{equation}
m\geq \sqrt{\frac{A_k}{16\pi}+\frac{Q^2}{2}+\frac{\pi(Q^4+4\mathcal{J}^2)}{A_k}},
\end{equation}
and equality occurs if and only if the initial data agree with that of the corresponding Kerr-Newman spacetime. In particular, if $A=A_k$ then the Penrose inequality with angular momentum and charge holds.
\end{cor}

This type of result may be considered a generalization of that of Gibbons and
Holzegel in \cite{GH}, who established the Penrose inequality without contributions from angular momentum and charge by utilizing the advantages of Weyl coordinates. In that paper they also had a more stringent condition than that of Corollary \ref{corollary}, concerning the agreement between the initial data and associated Schwarzschild solution on the horizon. Another related result is that of Chrusciel and Nguyen \cite{ChruscielNguyen} who utilize a related coordinate system referred to as pseudospherical coordinates, and obtain a mass lower bound in terms of the horizon rod length.

This paper is organized as follows. In Section \ref{sec2} we obtain the preliminary mass lower bound arising from Weyl coordinates, and relate it to a reduced harmonic energy. Section \ref{sec3} is dedicated to examining the various asymptotics of
relevant quantities in Weyl coordinates, and in Section \ref{sec4} it is established that the Kerr-Newman black hole is a global minimizer of the reduced harmonic energy. Finally, Section \ref{sec5} is dedicated to the proof of the main results.
Two appendices are included to discuss technical issues related to the metric coefficients in Weyl coordinates near the poles of the horizon, as well as the computations to show the relationship between $\overline{\beta}$ and surface gravity in the Kerr setting.

\section{The Mass Formula and Reduced Harmonic Energy}
\label{sec2}\setcounter{equation}{0}
\setcounter{section}{2}

An initial data set $(M,g,k)$ as in Theorem \ref{maintheorem} admits a global set of Weyl coordinates \cite{ChruscielNguyen} $(\rho,z,\phi)$ in which the metric takes the form \eqref{brill} and the scalar curvature is given by \cite{Brill}
\begin{equation}\label{scurvature}
2e^{-2U+2\alpha}R=8\Delta U-4\Delta_{\rho,z}\alpha-4|\nabla U|^{2}
-\rho^{2}e^{-2\alpha}\left(\partial_{z}A_{\rho}-\partial_{\rho}A_{z}\right)^{2},
\end{equation}
where $\Delta$ is the Laplacian with respect to the flat metric on $\mathbb{R}^3$ and
$\Delta_{\rho,z}=\partial_{\rho}^2+\partial_{z}^2$. Since there is a single black hole, or rather one minimal surface boundary component, the $z$-axis is broken up into three intervals or `rods' $(-\infty,-m_0)$, $(-m_0,m_0)$, $(m_0,\infty)$ in which the two semi-infinite rods are the axis and the finite rod represents the horizon boundary $\partial M$. The value $m_0>0$ is uniquely determined by the geometry of the initial data. Let $U_0$ and $\alpha_0$ denote the metric coefficients in Weyl coordinates for the Schwarzschild solution having this same rod structure; note that $m_0$ is then the mass
of this Schwarzschild spacetime. If $r_+=\sqrt{\rho^2+(z-m_0)^2}$ and $r_-=\sqrt{\rho^2+(z+m_0)^2}$ denote the Euclidean distances to the poles $p_+=(0,m_0)$ and $p_-=(0,-m_0)$ in the $\rho z$-plane, then
\begin{equation}\label{Ur}
U_0=\frac{1}{2} \log \frac{r_-+r_+-2m_0}{r_+ + r_- + 2m_0},\quad\quad
\alpha_0=\frac{1}{2} \log \frac{(r_- +r_+)^2-4m^2_0}{4r_- r_+}.
\end{equation}
These functions blow-up on the horizon but are finite along the axis.
In particular
\begin{equation}\label{l}
U_0=-\frac{m_0}{r}+O\left(\frac{1}{r^2}\right),\quad\quad
\alpha_0=O\left(\frac{1}{r^2}\right)\quad\text{ as }r:=\sqrt{\rho^2+z^2}\rightarrow\infty,
\end{equation}
\begin{equation}\label{ll}
U_0=\frac{1}{2}\log\left(\frac{z-m_0}{z+m_0}\right)+O(\rho^2),
\quad\quad \alpha_0=O(\rho^2)\quad\text{ as }\rho\rightarrow 0 \text{ and }
|z|\geq m_0+\epsilon,
\end{equation}
\begin{equation}\label{lll}
U_0=\log\rho+O(1),
\quad\quad \alpha_0=\log\rho+O(1)\quad\text{ as }\rho\rightarrow 0 \text{ and }
|z|\leq m_0-\epsilon,
\end{equation}
where $\epsilon>0$. These Schwarzschild
coefficients play the role of singular part for the metric coefficients of \eqref{brill}. That is, we may write $U=U_0 +\overline{U}$ and $\alpha=\alpha_0 +\overline{\alpha}$ where $\overline{U}$ and $\overline{\alpha}$ remain bounded. In fact, this decomposition has the following regularity properties which are proved in the appendix and rely on the minimal surface condition at the boundary.

\begin{lemma}\label{renormalized}
Under the assumptions of Theorem \ref{maintheorem} the renormalized
functions $\overline{U}$ and $\overline{\alpha}$ are smooth away from the
horizon rod, and have continuous first derivatives everywhere except possibly
at the poles $p_{\pm}$ where they are bounded.
At infinity $\overline{U}=O_{1}(r^{-1/2-\epsilon})$
and $\overline{\alpha}=O_{1}(r^{-1/2-\epsilon})$ for some $\epsilon>0$.
\end{lemma}

Let us now use this decomposition of the metric coefficients to compute the ADM mass. Recall from \cite{ChruscielNguyen} that if $S_{\infty}$ represents the limit as $r\rightarrow\infty$ for coordinate spheres $S_r$ then the mass is given by
\begin{equation}\label{MASS}
m=\frac{1}{8\pi}
\int_{S_{\infty}}\left[\partial_{r}(2U-\alpha)+\frac{\alpha}{r}\right]d\sigma.
\end{equation}
The boundary terms at infinity in this formula arise from integrating the scalar curvature formula \eqref{scurvature}. Observe that
\begin{align}\label{5}
\begin{split}
\int_{\mathbb{R}^{3}}\Delta_{\rho,z}\alpha dx=&
\int_{\mathbb{R}^{2}_{+}}2\pi\rho\Delta_{\rho,z}\alpha d\rho dz\\
=&\lim_{\varepsilon\rightarrow 0}\int_{\rho=\varepsilon}2\pi(\alpha-\rho\partial_{\rho}\alpha)dz
+\lim_{r\rightarrow\infty}\int_{\partial D_{r}^{+}}2\pi(\rho\partial_{r}\alpha-\alpha\sin\theta)ds\\
=&\lim_{\varepsilon\rightarrow 0}\int_{\rho=\varepsilon}2\pi(\alpha-\rho\partial_{\rho}\alpha)dz
+\int_{S_{\infty}}\left(\partial_r \alpha-\frac{\alpha}{r}\right)d\sigma.
\end{split}
\end{align}
Here $D_{r}^{+}$ is the half disk of radius $r$, and $\rho=r\sin\theta$ and $z=r\cos\theta$. Furthermore
\begin{equation}\label{6}
\int_{\mathbb{R}^{3}}\Delta U dx=\int_{S_{\infty}}\partial_{r}U d\sigma
-\lim_{\varepsilon\rightarrow 0}\int_{\rho=\varepsilon}2\pi\rho\partial_{\rho}U dz,
\end{equation}
and since $U_0=O_1(r^{-1})$ as $r\rightarrow\infty$ with $U_0$ harmonic
\begin{align}\label{7}
\begin{split}
\int_{\mathbb{R}^{3}}|\nabla U|^{2}dx=&\int_{\mathbb{R}^{3}}|\nabla(U_{0}+\overline{U})|^{2}dx\\
=&\int_{\mathbb{R}^{3}}\left(|\nabla \overline{U}|^{2}
+\nabla (U_{0}+2\overline{U})\cdot\nabla U_{0}\right)dx\\
=&\int_{\mathbb{R}^{3}}|\nabla \overline{U}|^{2}dx
-\lim_{\varepsilon\rightarrow 0}\int_{\rho=\varepsilon}(U_{0}+2\overline{U})\partial_{\rho}U_{0}d\sigma
+\int_{S_{\infty}}(U_{0}+2\overline{U})\partial_{r}U_{0}d\sigma\\
=&\int_{\mathbb{R}^{3}}|\nabla \overline{U}|^{2}dx
-\lim_{\varepsilon\rightarrow 0}\int_{\rho=\varepsilon}2\pi\rho(U_{0}+2\overline{U})\partial_{\rho}U_{0}dz.
\end{split}
\end{align}
Therefore by integrating the scalar curvature formula, and putting all these computations together, we find that
\begin{align}\label{8}
\begin{split}
8\pi m=&\int_{\mathbb{R}^{3}}\left[|\nabla \overline{U}|^{2}
+\frac{1}{2}e^{-2U+2\alpha}R+\frac{1}{4}
\rho^2 e^{-2\alpha}(\partial_z A_\rho-\partial_\rho A_z)^2\right]dx\\
&+\lim_{\varepsilon\rightarrow 0} \int_{\rho=\varepsilon}\left[4\pi\rho\partial_{\rho}U
+2\pi(\alpha-\rho\partial_{\rho}\alpha)
-2\pi\rho(U_{0}+2\overline{U})\partial_{\rho}U_{0}\right]dz.
\end{split}
\end{align}

Consider now the boundary integrals in \eqref{8}. Computations show that
\begin{equation}
\lim_{\varepsilon\rightarrow 0}\int_{\rho=\varepsilon}
 \rho \partial_{\rho}U dz
=\lim_{\varepsilon\rightarrow 0}\int_{\rho=\varepsilon}
\varepsilon \partial_{\rho}U_0(\varepsilon, z)dz=2 m_0,
\end{equation}
and
\begin{align}
\begin{split}
\lim_{\varepsilon\rightarrow 0}\int_{\rho=\varepsilon}
\left[\alpha-\rho\partial_{\rho}\alpha
-\rho(U_0+2\overline{U})\partial_{\rho}U_0\right] dz
=&\lim_{\varepsilon\rightarrow 0}
\int_{\rho=\varepsilon \atop |z|<m_0}\left(\alpha_0+\overline{\alpha}-\rho\partial_{\rho}\alpha_0
-U_0-2\overline{U}\right)dz\\
=&\int_{-m_0}^{m_0}(\overline{\alpha}-2\overline{U})(0,z)dz.
\end{split}
\end{align}
Furthermore, simple connectedness and the divergence free condition for the electric and magnetic fields gives rise to electromagnetic potentials
\cite[Section 2]{KWcharge}
\begin{equation}
d\psi =F(\eta,\cdot),\quad\quad\quad
d\chi =\star F(\eta,\cdot),
\end{equation}
where $F$ is the field strength tensor and $\star$ denotes the Hodge star operation. Similarly the compatibility condition $J_{em}(\eta)=0$ guarantees the existence of a charged twist potential
\begin{equation}
dv=k(\eta)\times \eta-\chi d\psi+\psi d\chi.
\end{equation}
Since the initial data are maximal, nonnegativity of the energy density $\mu_{em}\geq 0$ implies the following lower bound \cite[Section 2]{KWcharge} for scalar curvature
\begin{equation}
R\geq |k|^{2}_g+2(|E|^2_g+|B|^2_g)\geq
2\frac{e^{6U-2\alpha}}{\rho^{4}}|\nabla v+\chi\nabla\psi-\psi\nabla\chi|^{2}
+2\frac{e^{4U-2\alpha}}{\rho^{2}}(|\nabla\chi|^2+|\nabla\psi|^2).
\end{equation}
Putting all this together yields the mass lower bound
\begin{align}\label{26}
\begin{split}
m\geq&\frac{1}{8\pi}\int_{\mathbb{R}^{3}}\left(|\nabla \overline{U}|^{2}
+\frac{e^{4U}}{\rho^{4}}|\nabla v+\chi\nabla\psi-\psi\nabla\chi|^{2}
+\frac{e^{2U}}{\rho^{2}}(|\nabla\chi|^2+|\nabla\psi|^2)\right)dx\\
&+\frac{1}{4}\int_{-m_0}^{m_0}(\overline{\alpha}(0,z)-2\overline{U}(0,z))dz
+m_0.
\end{split}
\end{align}
Related formulas were obtained in \cite{Chrusciel,ChruscielNguyen} and \cite{GH} in different settings.

The volume integral on the right-hand side of \eqref{26} is directly related to the harmonic energy of maps between $\mathbb{R}^3\setminus\Gamma\rightarrow \mathbb{H}^{2}_{\mathbb{C}}$, where $\Gamma=\{\rho=0, |z|>m_0\}$ is the axis. More precisely, let $\tilde{\Psi}=(u,v,\chi,\psi):\mathbb{R}^{3}\setminus\Gamma\rightarrow
\mathbb{H}_{\mathbb{C}}^{2}$ and consider the harmonic energy of this map on a bounded domain $\Omega\subset\mathbb{R}^{3}\setminus\Gamma$:
\begin{equation}
E_{\Omega}(\tilde{\Psi})=\int_{\Omega}|\nabla u|^{2}+e^{4u}|\nabla v
+\chi\nabla\psi-\psi\nabla\chi|^{2}
+e^{2u}\left(|\nabla\chi|^{2}
+|\nabla\psi|^{2}\right)dx.
\end{equation}
Set $u=U-\log\rho$, then the reduced
energy $\mathcal{I}_{\Omega}$ of the renormalized map $\Psi=(\overline{U},v,\chi,\psi)$ is related to the harmonic energy of $\tilde{\Psi}$ by
\begin{equation}\label{51}
\mathcal{I}_{\Omega}(\Psi)=E_{\Omega}(\tilde{\Psi})
+\int_{\partial\Omega}(2\overline{U}+U_0-\log\rho)\partial_{\nu}(\log\rho-U_0)d\sigma,
\end{equation}
where $\nu$ denotes the unit outer normal to the boundary $\partial\Omega$ and
\begin{equation}
\mathcal{I}_{\Omega}(\Psi)=
\int_{\Omega}|\nabla \overline{U}|^{2}+\frac{e^{4U}}{\rho^{4}}|\nabla v
+\chi\nabla\psi-\psi\nabla\chi|^{2}
+\frac{e^{2U}}{\rho^{2}}\left(|\nabla\chi|^{2}
+|\nabla\psi|^{2}\right)dx.
\end{equation}
Observe that the volume integral of \eqref{26} is exactly the reduced energy
on $\mathbb{R}^3$, which will be denoted by $\mathcal{I}(\Psi)$.
The relation \eqref{51} is established through an integration by parts, using the fact that $\log\rho$ and $U_{0}$ are harmonic on $\mathbb{R}^{3}\setminus\Gamma$. Namely
\begin{align}
\begin{split}
\mathcal{I}_{\Omega}(\Psi)=&\int_{\Omega}\left(|\nabla(u-U_{0}+\log\rho)|^{2}
+e^{4u}|\nabla v+\chi \nabla \psi-\psi\nabla\chi|^{2}+e^{2u}(|\nabla \chi|^2+|\nabla \psi|^2)\right)dx\\
=&\int_{\Omega}|\nabla u|^{2}+\nabla(2u-U_{0}+\log\rho)\cdot\nabla(\log\rho-U_{0})dx\\
&+\int_{\Omega}e^{4u}|\nabla v+\chi \nabla \psi-\psi\nabla\chi|^{2}+e^{2u}(|\nabla \chi|^2+|\nabla \psi|^2) dx\\
=&\int_{\Omega}\left(|\nabla u|^{2}+e^{4u}|\nabla v+\chi \nabla \psi-\psi\nabla\chi|^{2}+e^{2u}(|\nabla \chi|^2+|\nabla \psi|^2)\right)dx\\
&+\int_{\partial\Omega}(2u-U_{0}+\log\rho)\partial_{\nu}(\log\rho-U_{0})d\sigma\\
=&E_{\Omega}(\tilde{\Psi})+\int_{\partial\Omega}(2\overline{U}
+U_{0}-\log\rho)\partial_{\nu}(\log\rho-U_{0})d\sigma.
\end{split}
\end{align}
The functional $\mathcal{I}$ may be considered a regularization of $E$ since the infinite term $\int|\nabla(\log\rho-U_{0})|^{2}$ has been removed,
and since the two functionals differ only by a boundary term they must have the same critical points.

Let $\tilde{\Psi}_{k}=(u_{k},v_{k},\chi_{k},\psi_{k})$ denote the harmonic map associated with the Kerr-Newman solution, and
let $\Psi_{k}$ be the corresponding renormalized map where $u_{k}=U_{k}-\log\rho
=\overline{U}_{k}+U_{0}-\log\rho$. It follows that
$\Psi_{k}$ is a critical point of $\mathcal{I}$. As will be shown in Section \ref{sec4}, $\Psi_{k}$ realizes the global minimum for $\mathcal{I}$.

\begin{theorem}\label{inf}
Suppose that $\Psi=(\overline{U},v,\chi, \psi)$ is smooth and satisfies the asymptotics \eqref{055}-\eqref{062}.
If $v|_{\Gamma}=v_{k}|_{\Gamma}$, $\chi|_{\Gamma}=\chi_{k}|_{\Gamma}$, and $\psi|_{\Gamma}=\psi_{k}|_{\Gamma}$ then there
exists a constant $C>0$ such that
\begin{equation}\label{10293}
\mathcal{I}(\Psi)-\mathcal{I}(\Psi_{k})
\geq C\left(\int_{\mathbb{R}^{3}}
\operatorname{dist}_{\mathbb{H}_{\mathbb{C}}^{2}}^{6}(\Psi,\Psi_{k})dx
\right)^{\frac{1}{3}}.
\end{equation}
\end{theorem}

This is a key result in the proof of the main theorem. Inequality \eqref{10293} together with the mass formula \eqref{26} yield a lower bound for the ADM mass in terms of the reduced energy of the unique Kerr-Newman harmonic map possessing the same angular momentum, charge, and horizon rod length as the given initial data. Since this Kerr-Newman harmonic energy is computed to have the correct expression for the Penrose inequality, the desired result \eqref{4123} follows. Details of the proof are given in Section \ref{sec5}.

\section{Asymptotics in Weyl Coordinates}
\label{sec3}\setcounter{equation}{0}
\setcounter{section}{3}

In order to minimize the functional $\mathcal{I}(\Psi)$ it is necessary to
choose the appropriate asymptotics for the map $\Psi$. The asymptotics will be
guided by the principle of having a finite reduced energy, however the convexity minimization argument of the next section will in general require stronger asymptotics than that which is optimal for integrability. It will be useful to first record the asymptotics of the Schwarzschild metric coefficients near the poles, namely a computation shows that
\begin{equation}\label{52}
e^{U_0}=O(r_+^{1/2})\quad\text{ as }r_+ \rightarrow 0\text{ and }
z\geq m_0,\quad\quad e^{U_0}=O(\rho r_+^{-1/2})\quad\text{ as }r_+ \rightarrow 0\text{ and }z\leq m_0,
\end{equation}
\begin{equation}\label{53}
e^{U_0}=O(\rho r_-^{-1/2})\quad\text{ as }r_- \rightarrow 0\text{ and }
z\geq -m_0,\quad\quad e^{U_0}=O(r_-^{1/2})\quad\text{ as }r_- \rightarrow 0\text{ and }z\leq -m_0,
\end{equation}
\begin{equation}\label{54}
e^{U_0-\alpha_0}=O(r_{\pm}^{1/2})\quad\quad\quad\text{ as }r_{\pm} \rightarrow 0.
\end{equation}

According to Lemma \ref{renormalized} we have
\begin{equation}\label{055}
\overline{U}\in C^{0,1}(\mathbb{R}^3),\quad\quad \overline{U}=O_{1}(r^{-1/2-\epsilon})\quad\text{ as }r\rightarrow\infty,
\end{equation}
which is enough to guarantee that the first term of $\mathcal{I}(\Psi)$ is finite.
Consider now the potential terms and set
$\omega=dv+\chi d\psi-\psi d\chi$. In order to achieve integrability at infinity
and near the axes away from the poles we will require, for $\lambda > \frac{3}{2}$, the following asymptotics
\begin{equation}\label{056}
|\omega|=\rho^{2}O(r^{-\lambda}),\text{ }\text{ }\text{ }\text{ }|\nabla\chi|+|\nabla\psi|=\rho O(r^{-\lambda})\quad\quad\text{ as }r\rightarrow\infty,
\end{equation}
\begin{equation}\label{057}
|\omega|=O(\rho^{2}),\text{ }\text{ }\text{ }\text{ }|\nabla\chi|+|\nabla\psi|=O(\rho) \quad\quad\text{ as }\rho\rightarrow 0\text{ and }|z|> m_0,
\end{equation}
\begin{equation}\label{058}
|\chi|,|\psi|=\text{const}+\rho^{2}O(r^{-\lambda})\quad\quad\text{ as }r\rightarrow\infty,
\end{equation}
\begin{equation}\label{059}
|\chi|,|\psi|=\text{const}+O(\rho^{2})\quad\quad\text{ as }\rho\rightarrow 0\text{ and }|z|> m_0,
\end{equation}
from which it follows that
\begin{equation}\label{060}
|\nabla v|=\rho O(r^{-\lambda+1})\quad\text{ as }r\rightarrow\infty,
\quad\quad
|\nabla v|= O(\rho)\quad\text{ as }\rho\rightarrow 0\text{ and }|z|> m_0.
\end{equation}

It remains to prescribe asymptotics near the poles and in a neighborhood of the
horizon rod. By \eqref{52}, $e^{4U}=O(r_+^2)$ or $e^{4U}=O(\rho^4 r_+^{-2})$ near $p_+$ if $z\geq m_0$ or $z\leq m_0$ respectively. It follows that the second term in $\mathcal{I}(\Psi)$ is integrable near $p_+$ if
\begin{equation}\label{061.0}
|\omega|=\rho^2 O(r_+^{-3/2})\quad \text{ for } z\geq m_0,\quad\quad
|\omega|= O(r_+^{1/2})\quad \text{ for } z \leq m_0.
\end{equation}
Similarly, near $p_-$ we will impose
\begin{equation}\label{061.00}
|\omega|=O(r_-^{1/2})\quad \text{ for } z\geq -m_0,\quad\quad
|\omega|=\rho^2 O(r_-^{-3/2}) \quad \text{ for } z \leq -m_0.
\end{equation}
Analogous considerations lead to the condition near $p_+$
\begin{equation}\label{061}
|\nabla\chi|+|\nabla\psi|=\rho O(r_+^{-1})\quad \text{ for } z\geq m_0,\quad\quad
|\nabla\chi|+|\nabla\psi|= O(1)\quad \text{ for } z \leq m_0,
\end{equation}
and near $p_-$
\begin{equation}\label{061.1}
|\nabla\chi|+|\nabla\psi|=O(1)\quad \text{ for } z\geq -m_0,\quad\quad
|\nabla\chi|+|\nabla\psi|=\rho O(r_-^{-1})\quad \text{ for } z \leq -m_0.
\end{equation}
Next observe that since $e^U =O(\rho)$ near the interior of the horizon rod,
if
\begin{equation}\label{062}
|\omega|=|\nabla\chi|=|\nabla\psi|= O(1)\quad\quad\text{ as }\rho\rightarrow 0\text{ and }|z|< m_0,
\end{equation}
then the last two terms of the reduced energy are integrable in this region.

Lastly we record additional asymptotics that follow from above and will
be needed in the following section. Assuming that the value of the potentials on
the axes agree with those of the potentials for the Kerr-Newman map $\Psi_k$, we may integrate on lines perpendicular to the axes and near $p_{\pm}$ to obtain
\begin{equation}\label{063}
|v-v_{k}|+|\chi-\chi_{k}|+|\psi-\psi_{k}|=O(\rho^2 r_{\pm}^{-1})\quad\quad
\text{ as }r_{\pm}\rightarrow 0\text{ and }|z|\geq m_0.
\end{equation}
For $|z|\leq m_0$, integrating on horizontal lines will not yield such an estimate
since the two sets of potentials do not necessarily agree on the horizon rod. Thus, we integrate along radial lines emanating from the poles $p_{\pm}$ to find
\begin{equation}\label{064}
|v-v_{k}|+|\chi-\chi_{k}|+|\psi-\psi_{k}|=O(r_{\pm})
\quad\quad
\text{ as }r_{\pm}\rightarrow 0\text{ and }|z|\leq m_0.
\end{equation}




\section{Minimizing the Functional}
\label{sec4} \setcounter{equation}{0}
\setcounter{section}{4}

In this section it will be shown that the renormalized Kerr-Newman harmonic map
$\Psi_k$ is the global minimizer of the functional $\mathcal{I}$, among competitors
$\Psi$ satisfying the asymptotics of Section \ref{sec3}. This is based on the convexity of
harmonic energy $E$ for nonpositively curved target spaces under geodesic deformations. Such a strategy has been used successfully in connection with mass-angular momentum-charge inequalities in \cite{ChruscielLiWeinstein,KWcharge,SchoenZhou}, where the minimizer arises from extreme black holes. Here we will extend this method to the setting of nondegenerate black holes. The difficulty arises from the fact that the convexity property does not pass directly from $E$ to $\mathcal{I}$ since the energy is applied to singular maps. To get around this problem a cut-and-paste procedure is employed in which the
regularized map $\Psi$ is approximated by maps $\Psi_{\delta,\varepsilon}$ which
agree with $\Psi_k$ on certain domains.
More precisely, let $\delta,\varepsilon>0$ be
small parameters and set $\Omega_{\delta,\varepsilon}=\{\delta<r_{\pm}; \text{ }r<2/\delta;\text{ }
\rho>\varepsilon\}$ and $\mathcal{A}_{\delta,\varepsilon}=B_{2/\delta}\setminus
\Omega_{\delta,\varepsilon}$, where $B_{2/\delta}$ is the coordinate ball of radius $2/\delta$. Then $\Psi_{\delta,\varepsilon}=(\overline{U}_{\delta,\varepsilon},v_{\delta,\varepsilon},
\chi_{\delta,\varepsilon},\psi_{\delta,\varepsilon})$ will be constructed so that
\begin{equation}\label{454}
\operatorname{supp}(\overline{U}_{\delta,\varepsilon}-\overline{U}_{k})\subset B_{2/\delta},\text{ }\text{ }\text{ }\text{ }\text{ }
\operatorname{supp}(v_{\delta,\varepsilon}-v_{k},\chi_{\delta,\varepsilon}-\chi_{k},
\psi_{\delta,\varepsilon}-\psi_{k})\subset \Omega_{\delta,\varepsilon}.
\end{equation}
If $\tilde{\Psi}^{t}_{\delta,\varepsilon}$, $t\in[0,1]$ is a geodesic in $\mathbb{H}_{\mathbb{C}}^{2}$ connecting
$\tilde{\Psi}^{1}_{\delta,\varepsilon}=\tilde{\Psi}_{\delta,\varepsilon}$ and $\tilde{\Psi}^{0}_{\delta,\varepsilon}=\tilde{\Psi}_{k}$, then $\tilde{\Psi}^{t}_{\delta,\varepsilon}\equiv\Psi_{k}$ outside $B_{2/\delta}$ and
$v^{t}_{\delta,\varepsilon}=v_{k}, \chi^t_{\delta,\varepsilon}=\chi_{k}$, and $\psi^t_{\delta,\varepsilon}=\psi_k$
on a neighborhood of $\mathcal{A}_{\delta,\varepsilon}$. We then have that $\overline{U}^{t}_{\delta,\varepsilon}
=\overline{U}_{k}+t(\overline{U}_{\delta,\varepsilon}-\overline{U}_{k})$ on this domain. The fact that this expression is linear in $t$, together with convexity of the harmonic energy produces
\begin{equation}\label{066}
\frac{d^{2}}{dt^{2}}\mathcal{I}(\Psi^{t}_{\delta,\varepsilon})
\geq 2\int_{\mathbb{R}^{3}}|\nabla\operatorname{dist}_{\mathbb{H}_{\mathbb{C}}^{2}}
(\Psi_{\delta,\varepsilon},\Psi_{k})|^{2}dx.
\end{equation}
Furthermore, since $\Psi_{k}$ is a critical point it follows that
\begin{equation}\label{456}
\frac{d}{dt}\mathcal{I}(\Psi^{t}_{\delta,\varepsilon})|_{t=0}=0.
\end{equation}
The gap bound of Theorem \ref{inf} is then obtained by integrating \eqref{066}, applying a Sobolev inequality, and taking the limit as $\delta,\varepsilon\rightarrow 0$. Each of these steps will now be justified.
Repeated use of the asymptotics in Section \ref{sec3} will be made, sometimes implicitly without reference to a particular equation.

The following cut-off functions are needed to construct the approximations
$\Psi_{\delta,\varepsilon}$. Namely
\begin{equation}\label{068}
\varphi_{\delta}=\begin{cases}
0 & \text{ if $r_{\pm}\leq\delta$,} \\
|\nabla\varphi_{\delta}|\leq \frac{2}{\delta} &
\text{ if $\delta<r_{\pm}<2\delta$,} \\
1 & \text{ if $r_{\pm}\geq2\delta$,} \\
\end{cases}
\end{equation}
\begin{equation}\label{069}
\varphi_{\delta}^{1}=\begin{cases}
1 & \text{ if $r\leq\frac{1}{\delta}$,} \\
|\nabla\varphi_{\delta}^{1}|\leq 2\delta &
\text{ if $\frac{1}{\delta}<r<\frac{2}{\delta}$,} \\
0 & \text{ if $r\geq\frac{2}{\delta}$,} \\
\end{cases}
\end{equation}
\begin{equation}\label{070}
\phi_{\varepsilon}=\begin{cases}
0 & \text{ if $\rho\leq\varepsilon$,} \\
\frac{\log(\rho/\varepsilon)}{\log(\sqrt{\varepsilon}/
\varepsilon)} &
\text{ if $\varepsilon<\rho<\sqrt{\varepsilon}$,} \\
1 & \text{ if $\rho\geq\sqrt{\varepsilon}$.} \\
\end{cases}
\end{equation}
The first step deals with neighborhoods of the poles $p_{\pm}$.
Let
$\mathcal{F}_{\delta}(\Psi)=(\overline{U},v_{\delta},\chi_{\delta},\psi_{\delta})$
where
\begin{equation}\label{071}
(v_{\delta},\chi_{\delta},\psi_{\delta})=
(v_{k},\chi_{k},\psi_{k})
+\varphi_{\delta}(v-v_{k},\chi-\chi_{k},\psi-\psi_{k}),
\end{equation}
so that the potentials of $\mathcal{F}_{\delta}(\Psi)$ and $\Psi_{k}$ agree on $B_\delta(p_+)\cup B_\delta(p_-)$.

\begin{lemma}\label{noc2}
Suppose that $\Psi\equiv\Psi_{k}$ outside $B_{2/\delta}$, then $\lim_{\delta\rightarrow 0}\mathcal{I}(\mathcal{F}_{\delta}(\Psi))=\mathcal{I}(\Psi)$.
\end{lemma}

\begin{proof}
Write
\begin{equation}\label{072}
\mathcal{I}(\mathcal{F}_{\delta}(\Psi))
=\sum_{\pm}\left[\mathcal{I}_{r_{\pm}<\delta}(\mathcal{F}_{\delta}(\Psi))
+\mathcal{I}_{\delta< r_{\pm}<2\delta}(\mathcal{F}_{\delta}(\Psi))
\right]+\mathcal{I}_{r_{\pm}>2\delta}(\mathcal{F}_{\delta}(\Psi)),
\end{equation}
where $r_{\pm}>2\delta$ denotes the complement of $B_{2\delta}(p_+)\cup B_{2\delta}(p_-)$. Then according to the dominated convergence theorem (DCT)
\begin{equation}\label{073}
\mathcal{I}_{r_{\pm}\geq2\delta}(\mathcal{F}_{\delta}(\Psi))
=\mathcal{I}_{r_{\pm}\geq2\delta}(\Psi)
\rightarrow \mathcal{I}(\Psi).
\end{equation}
Furthermore since the potentials of $\mathcal{F}_{\delta}(\Psi)$ and $\Psi_k$ agree on $r_{\pm}<\delta$, and $e^{U}\leq c e^{U_{k}}$
as $|\overline{U}|$ and $|\overline{U}_k|$ are bounded near $p_{\pm}$, the second and third integrands of
$\mathcal{I}_{r_{\pm}<\delta}(\mathcal{F}_{\delta}(\Psi))$ converge to zero in light of the finite reduced energy of $\Psi_k$. The first integrand involving $|\nabla\overline{U}|$ also tends to zero since this function remains bounded.

Now consider
\begin{equation}\label{074}
\mathcal{I}_{\delta< r_{\pm}<2\delta}(\mathcal{F}_{\delta}(\Psi))
=\underbrace{\int_{\delta< r_{\pm}<2\delta}|\nabla \overline{U}|^{2}}_{I_{1}}
+\underbrace{\int_{\delta< r_{\pm}<2\delta}
\frac{e^{4U}}{\rho^{4}}|\omega_{\delta}|^{2}}_{I_{2}}
+\underbrace{\int_{\delta< r_{\pm}<2\delta}
\frac{e^{2U}}{\rho^{2}}(|\nabla\chi_{\delta}|^{2}
+|\nabla\psi_{\delta}|^{2})}_{I_{3}},
\end{equation}
and note that $I_{1}\rightarrow 0$ by the DCT.
Next compute
\begin{align}\label{075}
\begin{split}
\omega_{\delta}=&\varphi_{\delta}\omega
+(1-\varphi_{\delta})\omega_{k}
+(v-v_{k})\nabla\varphi_{\delta}
+(\chi_{k}\psi-\psi_{k}\chi)\nabla\varphi_{\delta}
\\
&+\varphi_{\delta}(1-\varphi_{\delta})
[(\psi-\psi_{k})\nabla(\chi-\chi_{k})
-(\chi-\chi_{k})\nabla(\psi-\psi_{k})],
\end{split}
\end{align}
and use properties of the cut-off function to find
\begin{align}\label{076}
\begin{split}
I_{2}\leq& C\int_{\delta<r_{\pm}<2\delta}
\left(\frac{e^{4U}}{\rho^{4}}|\omega|^{2}
+\frac{e^{4U_{k}}}{\rho^{4}}|\omega_{k}|^{2}
+\frac{e^{4U}}{r_{\pm}^{2}\rho^{4}}|v-v_{k}|^{2}
+\frac{e^{4U}}{r_{\pm}^{2}\rho^{4}}|\chi_{k}\psi-\psi_{k}\chi|^{2}\right)\\
&+C\int_{\delta<r_{\pm}<2\delta}
\frac{e^{4U}}{\rho^{4}}\left(|\psi-\psi_{k}|^{2}|\nabla(\chi-\chi_{k})|^{2}
+|\chi-\chi_{k}|^{2}|\nabla(\psi-\psi_{k})|^{2}\right).
\end{split}
\end{align}
The first and second terms converge to zero by the DCT and finite reduced energies of $\Psi$ and $\Psi_{k}$. The third term may be estimated with the help of \eqref{063} and \eqref{064}, namely
\begin{equation}\label{077}
\int_{\delta<r_{\pm}<2\delta}
\frac{e^{4U}}{r_{\pm}^{2}\rho^{4}}|v-v_{k}|^{2}\leq \int_{\delta<r_{\pm}<2\delta}
Cr^{-2} \rightarrow 0,
\end{equation}
and similar considerations apply for the fourth term. For the fifth term employ \eqref{061}, \eqref{061.1}, \eqref{063}, and \eqref{064} to find
\begin{equation}
\int_{\delta<r_{\pm}<2\delta}
\frac{e^{4U}}{\rho^{4}}|\psi-\psi_{k}|^{2}|\nabla(\chi-\chi_{k})|^{2}
\leq \int_{\delta<r_{\pm}<2\delta}C\rightarrow 0,
\end{equation}
and similarly for the sixth term. This shows that $I_2\rightarrow 0$. Lastly,
analogous reasoning yields $I_3\rightarrow 0$.
\end{proof}

Consider now the asymptotically flat end and set
\begin{equation}
\mathcal{F}_{\delta}^{1}(\Psi)=\Psi_{k}
+\varphi_{\delta}^{1}(\Psi-\Psi_{k}),
\end{equation}
so that $\mathcal{F}_{\delta}^{1}(\Psi)=\Psi_{k}$ on $\mathbb{R}^{3}\setminus B_{2/\delta}$. Then as is shown in \cite[Lemma 4.2]{KWcharge}
\begin{equation}\label{con1}
\lim_{\delta\rightarrow 0}\mathcal{I}(\mathcal{F}_{\delta}^{1}(\Psi))=\mathcal{I}(\Psi).
\end{equation}
Next we treat the cylindrical regions around the axis and horizon rod, and will make use of the domains
\begin{equation}
\mathcal{C}_{\delta,\varepsilon}=
\{\rho<\varepsilon; \text{ }\delta< r_{\pm}; \text{ } r< 2/\delta\},
\end{equation}
\begin{equation}
\mathcal{W}_{\delta,\varepsilon}^1=
\{\varepsilon< \rho<\sqrt{\varepsilon}; \text{ }\delta< r_{\pm}; \text{ }r\leq 2/\delta;\text{ }|z|> m\},
\end{equation}
\begin{equation}
\mathcal{W}_{\delta,\varepsilon}^2=
\{\varepsilon< \rho<\sqrt{\varepsilon};\text{ }\delta< r_{\pm}; \text{ }|z|< m\}.
\end{equation}
Let
$\mathcal{G}_{\varepsilon}(\Psi)=(\overline{U},v_{\varepsilon},
\chi_{\varepsilon},\psi_{\varepsilon})$
where
\begin{equation}\label{080}
(v_{\varepsilon},
\chi_{\varepsilon},\psi_{\varepsilon})=
(v_{k},\chi_{k},\psi_{k})
+\phi_{\varepsilon}(v-v_{k},\chi-\chi_{k},\psi-\psi_{k}),
\end{equation}
so that the potentials of $\mathcal{G}_{\varepsilon}(\Psi)$  and $\Psi_{k}$ agree on $\rho<\varepsilon$.

\begin{lemma}\label{noc3}
Fix $\delta>0$. Assume that the potentials of $\Psi$ and $\Psi_{k}$ agree on $B_{\delta}(p_+)\cup B_{\delta}(p_-)$, and $\Psi\equiv\Psi_k$ outside $B_{2/\delta}$, then
$\lim_{\varepsilon\rightarrow 0}\mathcal{I}(\mathcal{G}_{\varepsilon}(\Psi))=\mathcal{I}(\Psi)$.
\end{lemma}

\begin{proof}
Write
\begin{equation}\label{081}
\mathcal{I}(\mathcal{G}_{\varepsilon}(\Psi))
=\mathcal{I}_{\mathcal{C}_{\delta,\varepsilon}}(\mathcal{G}_{\varepsilon}(\Psi))
+\mathcal{I}_{\mathcal{W}_{\delta,\varepsilon}^1}(\mathcal{G}_{\varepsilon}(\Psi))
+\mathcal{I}_{\mathcal{W}_{\delta,\varepsilon}^2}(\mathcal{G}_{\varepsilon}(\Psi))
+\mathcal{I}_{\mathbb{R}^{3}\setminus(\mathcal{C}_{\delta,\varepsilon}
\cup\mathcal{W}_{\delta,\varepsilon}^1\cup\mathcal{W}_{\delta,\varepsilon}^2)}
(\mathcal{G}_{\varepsilon}(\Psi)).
\end{equation}
Since the potentials of
$\Psi$ and $\Psi_{k}$ agree on $B_{\delta}(p_{\pm})$, the DCT and finite reduced energy imply that
\begin{equation}\label{082}
\mathcal{I}_{\mathbb{R}^{3}\setminus(\mathcal{C}_{\delta,\varepsilon}
\cup\mathcal{W}_{\delta,\varepsilon}^1\cup\mathcal{W}_{\delta,\varepsilon}^2)}
(\mathcal{G}_{\varepsilon}(\Psi))
\rightarrow \mathcal{I}(\Psi).
\end{equation}
Furthermore since the potentials of $\mathcal{G}_{\varepsilon}(\Psi)$ and $\Psi_k$ agree on $\mathcal{C}_{\delta,\varepsilon}$, and $e^{U}\leq c e^{U_{k}}$
on this region, the second and third integrands of
$\mathcal{I}_{\mathcal{C}_{\delta,\varepsilon}}(\mathcal{G}_{\varepsilon}(\Psi))$ converge to zero in light of the finite reduced energy of $\Psi_k$. The first integrand involving $|\nabla\overline{U}|$ also tends to zero since this function remains bounded.

The domain $\mathcal{W}_{\delta,\varepsilon}^1$ concerns a neighborhood of the axis of rotation, and therefore $\mathcal{I}_{\mathcal{W}_{\delta,\varepsilon}^1}(\mathcal{G}_{\varepsilon}(\Psi))
\rightarrow 0$ according to Lemma 4.4 of \cite{KWcharge}. Now consider
\begin{equation}\label{083}
\mathcal{I}_{\mathcal{W}_{\delta,\varepsilon}^2}(\mathcal{G}_{\varepsilon}(\Psi))
=\underbrace{\int_{\mathcal{W}_{\delta,\varepsilon}^2}|\nabla \overline{U}|^{2}}_{I_{1}}
+\underbrace{\int_{\mathcal{W}_{\delta,\varepsilon}^2}
\frac{e^{4U}}{\rho^{4}}|\omega_{\varepsilon}|^{2}}_{I_{2}}
+\underbrace{\int_{\mathcal{W}_{\delta,\varepsilon}^2}
\frac{e^{2U}}{\rho^{2}}(|\nabla\chi_{\varepsilon}|^{2}
+|\nabla\psi_{\varepsilon}|^{2})}_{I_{3}},
\end{equation}
and notice that $I_{1}\rightarrow 0$ since $|\nabla\overline{U}|$ remains bounded.
Next observe that
\begin{align}\label{084}
\begin{split}
\omega_{\varepsilon}=&\phi_{\varepsilon}\omega
+(1-\phi_{\varepsilon})\omega_{k}
+(v-v_{k})\nabla\phi_{\varepsilon}
+(\chi_{k}\psi-\psi_{k}\chi)\nabla\phi_{\varepsilon}
\\
&+\phi_{\varepsilon}(1-\phi_{\varepsilon})
[(\psi-\psi_{k})\nabla(\chi-\chi_{k})
-(\chi-\chi_{k})\nabla(\psi-\psi_{k})].
\end{split}
\end{align}
The asymptotics of the cut-off function then yield
\begin{align}\label{085}
\begin{split}
I_{2}\leq& C\int_{\mathcal{W}_{\delta,\varepsilon}^2}
\left(\frac{e^{4U}}{\rho^4}|\omega|^{2}
+\frac{e^{4U_k}}{\rho^4}|\omega_{k}|^{2}
+(\log\varepsilon)^{-2}\rho^{-2}|v-v_{k}|^{2}
+(\log\varepsilon)^{-2}\rho^{-2}|\chi_{k}\psi-\psi_{k}\chi|^{2}\right)\\
&+C\int_{\mathcal{W}_{\delta,\varepsilon}^2}
\left(|\psi-\psi_{k}|^{2}|\nabla(\chi-\chi_{k})|^{2}
+|\chi-\chi_{k}|^{2}|\nabla(\psi-\psi_{k})|^{2}\right).
\end{split}
\end{align}
The first two terms converge to zero by the finite reduced energies. Furthermore according to \eqref{062}, $|v-v_k|=O(1)$ and thus
\begin{equation}\label{086}
\int_{\mathcal{W}_{\delta,\varepsilon}^2}
(\log\varepsilon)^{-2}|v-v_{k}|^{2}
\leq C\int_{\mathcal{W}_{\delta,\varepsilon}^2}
(\log\varepsilon)^{-2}\rho^{-2}= O\left((\log\varepsilon)^{-1}\right)\rightarrow 0.
\end{equation}
Analogous considerations may be used to treat the fourth term. Lastly, since
$|\psi-\psi_{k}|$ and $|\nabla(\chi-\chi_{k})|$ remain bounded the fifth term
tends to zero, and similarly for the sixth.
\end{proof}

We are now in a position to construct the appropriate approximation to $\Psi$ via the cut and paste operations by composition
\begin{equation}\label{4105}
\Psi_{\delta,\varepsilon}
=\mathcal{G}_{\varepsilon}\left(\mathcal{F}_{\delta}\left(
\mathcal{F}_{\delta}^{1}(\Psi)\right)\right).
\end{equation}
Then according to \eqref{con1} and Lemmas \ref{noc2} and \ref{noc3},
\begin{equation}\label{319}
\lim_{\delta\rightarrow 0}\lim_{\varepsilon\rightarrow 0}
\mathcal{I}(\Psi_{\delta,\varepsilon})=\mathcal{I}(\Psi).
\end{equation}

\noindent\textit{Proof of Theorem \ref{inf}.}
As in the introduction to this section let $\tilde{\Psi}^{t}_{\delta,\varepsilon}$ be the geodesic deformation connecting $\tilde{\Psi}_k$ to $\tilde{\Psi}_{\delta,\varepsilon}$. Due to the properties of the approximation the first component of the geodesic is
$\overline{U}^{t}_{\delta,\varepsilon}=\overline{U}_{k}
+t(\overline{U}_{\delta,\varepsilon}-\overline{U}_{k})$ on $\mathcal{A}_{\delta,\varepsilon}$, and in particular $\operatorname{dist}_{\mathbb{H}_{\mathbb{C}}^{2}}(\Psi_{\delta,\varepsilon},\Psi_{k})
=|\overline{U}_{\delta,\varepsilon}-\overline{U}_{k}|$ on this domain. These two observations, together with the asymptotics near the poles $p_{\pm}$ show that one may differentiate under the integral sign to directly compute the second variation
and find
\begin{equation}
\frac{d^{2}}{dt^{2}}\mathcal{I}_{\mathcal{A}_{\delta,\varepsilon}}(\Psi^{t}_{\delta,\varepsilon})
\geq\int_{\mathcal{A}_{\delta,\varepsilon}}2|\nabla(\overline{U}_{\delta,\varepsilon}
-\overline{U}_k)|^2=\int_{\mathcal{A}_{\delta,\varepsilon}}2|\nabla
\operatorname{dist}_{\mathbb{H}_{\mathbb{C}}^{2}}(\Psi_{\delta,\varepsilon},\Psi_{k})|^2.
\end{equation}
On the domain $\Omega_{\delta,\varepsilon}$, the relation \eqref{51} between reduced and harmonic energies may be used. Due to the linearity of $\overline{U}^{t}_{\delta,\varepsilon}$ in $t$, the boundary term of \eqref{51} vanishes when computing the second variation so that
\begin{equation}
\frac{d^{2}}{dt^{2}}\mathcal{I}_{\Omega_{\delta,\varepsilon}}(\Psi^{t}_{\delta,\varepsilon})
=\frac{d^{2}}{dt^{2}}E_{\Omega_{\delta,\varepsilon}}
(\tilde{\Psi}^{t}_{\delta,\varepsilon})
\geq\int_{\Omega_{\delta,\varepsilon}}2|\nabla
\operatorname{dist}_{\mathbb{H}_{\mathbb{C}}^{2}}(\Psi_{\delta,\varepsilon},\Psi_{k})|^2,
\end{equation}
where the inequality is obtained from the convexity of harmonic energy \cite{SchoenZhou}. Since $\Omega_{\delta,\varepsilon}$ and $\mathcal{A}_{\delta,\varepsilon}$ are complementary in $B_{2/\delta}$, and the geodesic deformation is constant outside of this large ball, it follows that \eqref{066} holds.

Next, let $\bar{\delta}<\delta$ and $\bar{\varepsilon}<\varepsilon$, and observe
that since $\Psi_k$ is a critical point
\begin{equation}\label{alsk}
\frac{d}{dt}\mathcal{I}_{\Omega_{\bar{\delta},\bar{\varepsilon}}}
(\Psi^{t}_{\delta,\varepsilon})|_{t=0}
=-\sum_{\pm}\int_{\partial B_{\bar{\delta}}(p_{\pm})}2(\overline{U}_{\delta,\varepsilon}-\overline{U}_{k})
\partial_{\nu}\overline{U}_{k}
-\int_{\partial\mathcal{C}_{\bar{\delta},\bar{\varepsilon}}}
2(\overline{U}_{\delta,\varepsilon}-\overline{U}_{k})\partial_{\nu}\overline{U}_{k},
\end{equation}
where $\nu$ is the unit normal pointing towards infinity. In addition, using the
constancy of the potentials and linearity of $\overline{U}^{t}_{\delta,\varepsilon}$ on $\mathcal{A}_{\bar{\delta},\bar{\varepsilon}}$ we find that
\begin{align}\label{qpwo}
\begin{split}
\frac{d}{dt}\mathcal{I}_{\mathcal{A}_{\bar{\delta},\bar{\varepsilon}}}
(\Psi^{t}_{\delta,\varepsilon})|_{t=0}=&\int_{\mathcal{A}_{\bar{\delta},\bar{\varepsilon}}}
2\nabla \overline{U}_{k}\cdot\nabla(\overline{U}_{\delta,\varepsilon}-\overline{U}_{k})
+4(\overline{U}_{\delta,\varepsilon}-\overline{U}_{k})\frac{e^{4U_k}}{\rho^{4}}
|\omega_{k}|^{2}\\
&+\int_{\mathcal{A}_{\bar{\delta},\bar{\varepsilon}}} 2(\overline{U}_{\delta,\varepsilon}-\overline{U}_{k})
\frac{e^{2U_k}}
{\rho^{2}}\left(|\nabla \chi_{k}|^{2}+|\nabla \psi_{k}|^{2}\right).
\end{split}
\end{align}
Since $|\overline{U}|+|\nabla\overline{U}|$ is uniformly bounded, \eqref{alsk} tends to zero as $\bar{\varepsilon}\rightarrow 0$ followed by $\bar{\delta}\rightarrow 0$, and the same holds for \eqref{qpwo} since it may be estimated by the reduced energy of $\Psi_k$ on $\mathcal{A}_{\bar{\delta},\bar{\varepsilon}}$.

We may now integrate \eqref{066} two times and use a Sobolev inequality to obtain
the inequality \eqref{10293} of Theorem \ref{inf} with $\Psi$ replaced by $\Psi_{\delta,\varepsilon}$. In light of \eqref{319}, the desired result follows by
taking the limits as $\varepsilon\rightarrow 0$ and then $\delta\rightarrow 0$.
\hfill\qedsymbol\medskip

\section{Proof of the Main Results}
\label{sec5} \setcounter{equation}{0}
\setcounter{section}{5}

We first show that under the assumptions of Theorem \ref{maintheorem}
the potentials and quantities arising from Weyl coordinates satisfy
the asymptotics stated in Section \ref{sec3}. Lemma \ref{renormalized} guarantees
that $\overline{U}$ behaves in a manner consistent with \eqref{055}.
Next, as is shown in \cite{KWcharge}
\begin{equation}
\frac{e^{6U-2\alpha}}{\rho^{4}}|\nabla v
+\chi\nabla\psi-\psi\nabla\chi|^{2}\leq |k|^2_g.
\end{equation}
Consider a domain near the poles $p_{\pm}$ with $|z|\geq m_0$, then using \eqref{52}-\eqref{54} we find that
\begin{equation}
|\nabla v
+\chi\nabla\psi-\psi\nabla\chi|=O(\rho^{2}e^{-2U}e^{-U+\alpha})=O(\rho^2 r_{\pm}^{-3/2}),
\end{equation}
since $|k|_{g}$ remains bounded. Similarly if $|z|\leq m_0$
\begin{equation}
|\nabla v
+\chi\nabla\psi-\psi\nabla\chi|=O(\rho^{2}e^{-2U}e^{-U+\alpha})=O(r_{\pm}^{1/2}),
\end{equation}
which confirms \eqref{061.0} and \eqref{061.00}. Near the horizon rod away from the poles, that is $|z|<m_0$, the asymptotics \eqref{lll} imply
\begin{equation}
|\nabla v
+\chi\nabla\psi-\psi\nabla\chi|=O(\rho^{2}e^{-2U}e^{-U+\alpha})=O(1),
\end{equation}
confirming part of \eqref{062}.

For the electromagnetic potentials recall that from \cite{KWcharge},
\begin{equation}
\frac{e^{4U-2\alpha}}{\rho^{2}}\left(|\nabla\chi|^{2}
+|\nabla\psi|^{2}\right)\leq |E|^2_g +|B|^2_g.
\end{equation}
Again the right-hand side is bounded near the poles, so for $|z|\geq m_0$ we have
\begin{equation}
|\nabla\chi|+|\nabla\psi|=O(\rho e^{-U}e^{-U+\alpha})=O(\rho r_{\pm}^{-1}),
\end{equation}
and for $|z|\leq m_0$
\begin{equation}
|\nabla\chi|+|\nabla\psi|=O(\rho e^{-U}e^{-U+\alpha})=O(1).
\end{equation}
This shows that \eqref{061} and \eqref{061.1} hold.
Analogously, near the horizon rod with $|z|<m_0$
\begin{equation}
|\nabla \chi|+|\nabla \psi|=O(1),
\end{equation}
which fulfills \eqref{062}. Furthermore the asymptotics in a neighborhood of the axis, \eqref{057} and \eqref{059}, may be obtained in similar fashion. Lastly, \eqref{056} and \eqref{058} follow from asymptotic flatness.

We are now in a position to establish Theorem \ref{maintheorem}. As shown above, the map $\Psi$ arising from the initial data satisfies the hypotheses of Theorem \ref{inf}. Therefore, together with \eqref{26} the following lower bound for the mass is achieved
\begin{equation}\label{mmmmm}
m\geq\frac{1}{8\pi}\mathcal{I}(\Psi_k)
+\frac{1}{4}\int_{-m_0}^{m_0}(\overline{\alpha}(0,z)-2\overline{U}(0,z))dz
+m_0.
\end{equation}
Let $m_k$ and $A_k$ denote the mass and horizon area of the Kerr-Newman solution
associated with the map $\Psi_k$. Then since the Kerr-Newman solution is known to
saturate the Penrose inequality
\begin{align}\label{4124}
\begin{split}
m_{k}=&\sqrt{\frac{A_{k}}{16\pi}+\frac{Q^2}{2}+\frac{\pi(Q^4+4\mathcal{J}^2)}{A_{k}}} \\
=&\frac{1}{8\pi}\mathcal{I}(\Psi_k)
+\frac{1}{4}\int_{-m_0}^{m_0}(\overline{\alpha}_{k}(0,z)-2\overline{U}_{k}(0,z))dz
+m_0.
\end{split}
\end{align}
It follows that
\begin{equation}\label{hhhh}
m\geq   \sqrt{\frac{A_{k}}{16\pi}+\frac{Q^2}{2}+\frac{\pi(Q^4+4\mathcal{J}^2)}{A_{k}}} +\frac{1}{4}\int_{-m_0}^{m_0}(\overline{\beta}(0,z)-\overline{\beta}_k(0,z))dz,
\end{equation}
which is the desired inequality. In the case that this inequality is saturated we must have $\Psi=\Psi_k$ by Theorem \ref{inf}. Several other quantities arising from the derivation of \eqref{26} vanish, from which it may be shown that the initial data $(M,g,k)$ agrees with that of the canonical slice of the Kerr-Newman spacetime; details are given in \cite[Section 2]{KWcharge}.

We will now establish Corollary \ref{corollary}. If $\overline{\beta}$ is constant on the horizon rod then
\begin{equation}
e^{\frac{1}{2m_0}\int_{-m_0}^{m_0}
\overline{\beta}(0,z)dz}
=\frac{1}{2m_0}\int_{-m_0}^{m_0}e^{\overline{\beta}(0,z)}dz=\frac{A}{16\pi m_0^2}.
\end{equation}
The same equality holds for $\beta$, $A$ replaced by $\beta_k$, $A_k$ since $\beta_k$ is also constant on the horizon. Therefore if we assume that $A\geq A_k$, then
\begin{equation}
\int_{-m_0}^{m_0}\overline{\beta}(0,z)dz\geq\int_{-m_0}^{m_0}
\overline{\beta}_k(0,z)dz,
\end{equation}
which together with \eqref{hhhh} yields the desired inequality. The case of equality here is treated as above.
\hfill\qedsymbol\medskip

\appendix



\section{Weyl Coordinates}
\label{sec7} \setcounter{equation}{0}

Here we prove Lemma \ref{renormalized}. In \cite{ChruscielNguyen} the existence of Weyl coordinates was established by first constructing so called
pseudospherical coordinates $(\rho_{s},z_{s},\phi)$, in which the initial data boundary $\partial M$ is represented by a semi-circle of radius $\frac{m_0}{2}$ about the origin in the $\rho_s z_s$-plane. This contrasts with Weyl coordinates in which the boundary takes the form of an interval on the $z$-axis in the orbit space. Pseudospherical coordinates are valid on the planar region $\mathbb{C}_{+}\setminus D_{m_0/2}=\{\rho_s+iz_s\mid \rho_{s}>0, r_s>m_0/2\}$, where $r_s^2=\rho_s^2+z_s^2$. In these coordinates the metric takes the standard `Brill' form
\begin{equation}
g=e^{-2U_{s}+2\alpha_{s}}(d\rho_{s}^{2}+dz_{s}^{2})+\rho_s^{2}e^{-2U_{s}}
(d\phi+A_{\rho_s} d\rho_{s}+A_{z_s} dz_{s})^{2}.
\end{equation}
This structure for the metric is preserved under any coordinate change of the plane which yields a conformal transformation, and Weyl coordinates are a particular example of this. The metric coefficients are axisymmetric, smooth up to the boundary in $\mathbb{C}_{+}\setminus D_{m_0/2}$ with $\alpha_s=0$ on the $z_s$-axis, and satisfy the fall-off
\begin{equation}
U_{s}=O_{1}(r_{s}^{-1/2-\epsilon}),\quad
\alpha_{s}=O_{1}(r_{s}^{-1/2-\epsilon}),
\quad A_{\rho_s}=O_{1}(r_{s}^{-3/2-\epsilon}),\quad
A_{z_s}=O_{1}(r_{s}^{-3/2-\epsilon}).
\end{equation}

Weyl coordinates $(\rho,z,\phi)$ are constructed from pseudospherical coordinates as follows. Define complex coordinates $\zeta_{s}=\rho_{s}+iz_{s}$
and $\zeta=\rho+iz$ and consider the holomorphic diffeomorphism $f:\mathbb{C}_+\setminus D_{m_0/2}\to\mathbb{C}_+$ given by
\begin{equation}
\zeta=f(\zeta_{s})=\zeta_{s}-\frac{m_0^{2}}{4\zeta_{s}}\quad\quad\quad\Rightarrow
\quad\quad\quad\rho=\frac{\rho_{s}(r_{s}^{2}-\frac{m_0^{2}}{4})}{r_{s}^{2}},\quad
z=\frac{z_{s}(r_{s}^{2}+\frac{m_0^{2}}{4})}{r_{s}^{2}}.
\end{equation}
Observe that
\begin{equation}
\frac{\partial{\zeta}}{\partial{\zeta_s}}=1+\frac{m_0^2}{4\zeta_s^2},
\end{equation}
which is smooth up to the boundary of $\mathbb{C}_+\setminus D_{m_0/2}$ and is nonzero except at the points $\zeta_s=\pm \frac{m_0}{2}i$. Thus by the inverse function theorem, the inverse transformation is holomorphic and has bounded derivatives away from the poles $\zeta=\pm m_0i$ of the horizon. Near these points we have
\begin{equation}
\left|\frac{\partial{\zeta}}{\partial{\zeta_s}}\right|\geq C^{-1}|\zeta_s \mp\frac{m_0}{2}i|\quad\quad\Rightarrow\quad\quad
\left|\frac{\partial{\zeta_s}}{\partial{\zeta}}\right|\leq \frac{C}{|\zeta_s \mp\frac{m_0}{2}i|}.
\end{equation}
In particular, all first derivatives of the real and imaginary parts admit
the bound
\begin{equation}\label{qeqe}
\left|\frac{\partial{\rho_s}}{\partial{\rho}}\right|+
\left|\frac{\partial{\rho_s}}{\partial{z}}\right|+
\left|\frac{\partial{z_s}}{\partial{\rho}}\right|+
\left|\frac{\partial{z_s}}{\partial{z}}\right| \leq \frac{C}{|\zeta_s \mp\frac{m_0}{2}i|}
\end{equation}
near the poles.

The relationship between $U$, $\alpha$ of Weyl coordinates and $U_s$, $\alpha_s$ of pseudospherical coordinates is given by \cite{ChruscielNguyen}
\begin{equation}
U(\rho,z)=U_{s}(\rho_{s},z_{s})-\log\frac{\rho_{s}}{\rho},\quad\quad
\alpha(\rho,z)=\alpha_{s}(\rho_s,z_s)
+\log\frac{|\zeta_{s}|^{2}-\frac{m_0^{2}}{4}}{|\zeta_{s}^{2}+\frac{m_0^{2}}{4}|}.
\end{equation}
Note that the second term on the right-hand side of both expressions depends only on the coordinate transformation. For the Schwarzschild solution
\begin{equation}
U_{s,0}(\rho_{s},z_{s})=-2\log\frac{2r_{s}+m_0}{2r_{s}},\quad\quad
\alpha_{s,0}(\rho_s,z_s)=0,
\end{equation}
and the expressions for the Schwarzschild data $U_0$ and $\alpha_0$ in Weyl coordinates may then be obtained from the above formulas.
We may then write $U=U_{0}+\overline{U}$ and $\alpha=\alpha_{0}+\overline{\alpha}$ where
\begin{equation}
\overline{U}(\rho,z):=U(\rho,z)-U_{0}(\rho,z)
=U_{s}(\rho_{s},z_{s})-U_{s,0}(\rho_{s},z_{s}),
\end{equation}
and
\begin{equation}
\overline{\alpha}(\rho,z):=\alpha(\rho,z)-\alpha_{0}(\rho,z)
=\alpha_{s}(\rho_{s},z_{s}).
\end{equation}
It immediately follows that $\overline{U}$ and $\overline{\alpha}$ are uniformly
bounded and satisfy the desired decay at infinity.
Furthermore since $U_s$, $U_{s,0}$, and $\alpha_s$ are smooth, the regularity properties of $\overline{U}$ and $\overline{\alpha}$ depend on the coordinate transformation $f^{-1}$, and the only possible issues arise at the poles.

Consider the partial derivative
\begin{equation}\label{gfgf}
\frac{\partial\overline{U}}{\partial\rho}=\left(\frac{\partial U_{s}}{\partial\rho_{s}}-\frac{\partial U_{s,0}}{\partial\rho_{s}}\right)\frac{\partial{\rho_{s}}}{\partial{\rho}}
+\left(\frac{\partial U_{s}}{\partial{z_{s}}}-\frac{\partial U_{s,0}}{\partial{z_{s}}}\right)\frac{\partial{z_{s}}}{\partial{\rho}}.
\end{equation}
Since the horizon is a minimal surface
\begin{equation}
\frac{\partial}{\partial r_{s}}(U_{s}-\frac{1}{2}\alpha_{s})
=\frac{2}{m_0}=\frac{\partial U_{s,0}}{\partial r_{s}}\quad\quad
\text{ when }r_s=\frac{m_0}{2}.
\end{equation}
In particular this holds at $(\rho_{s},z_{s})=(0,\pm m_0/2)$. Moreover,
since $\alpha_{s}=0$ on the axis and $\partial_{r_s}$ coincides with $\pm\partial_{z_{s}}$ there, we have
\begin{equation}
\left(\frac{\partial U_{s}}{\partial z_{s}}-\frac{\partial U_{s,0}}{\partial z_{s}}\right)\left(0,\pm \frac{m_0}{2}\right)=0.
\end{equation}
Next, use the fact that all functions are axisymmetric to find
\begin{equation}
\frac{\partial U_{s}}{\partial\rho_{s}}\left(0,\pm\frac{m_0}{2}\right)
=\frac{\partial U_{s,0}}{\partial\rho_{s}}\left(0,\pm\frac{m_0}{2}\right)=0.
\end{equation}
Therefore the first derivatives of $U_{s}-U_{s,0}$ vanish at the poles. This, combined with the smoothness of this function up to the boundary, shows that
even though $\partial_{\rho}\rho_s$ and $\partial_{\rho}z_s$ may blow-up at these points in a manner controlled by \eqref{qeqe}, the full expression \eqref{gfgf} remains bounded. Similar considerations may be used to treat the $\partial_z \overline{U}$ and the derivatives of $\overline{\alpha}$.

\section{Relation of $\overline{\beta}$ to Surface Gravity}
\label{sec8} \setcounter{equation}{0}

Here we compute $\overline{\beta}=\overline{\alpha}-2\overline{U}$ on the horizon rod for the Kerr black hole. Let us recall the constant time slice Kerr metric $g_{kerr}$ in Weyl coordinates \cite{Stephani}. We will denote the mass and angular momentum of the Kerr metric by $m$ and $\mathcal{J}=ma$, while the notation for half the horizon rod length will be $m_0$. Then
\begin{equation}
g_{kerr}=e^{-2U_{kerr}+2\alpha_{kerr}}(d\rho^{2}+dz^{2})+\rho^{2}e^{-2U_{kerr}}d\phi^{2},
\end{equation}
where
\begin{equation}
e^{-2U_{kerr}+2\alpha_{kerr}}
=\frac{m_0^{2}\left(r_{+}+r_{-}+2m\right)^{2}+a^{2}(r_{+}-r_{-})^{2}}
{4m_0^{2}r_{+}r_{-}},
\end{equation}
\begin{align}
\begin{split}
&\rho^{2}e^{-2U_{kerr}}\\
=&
\frac{m_0^{2}\left(r_{+}+r_{-}+2m\right)^{2}+a^{2}(r_{+}-r_{-})^{2}}
{m_0^{2}\left((r_{+}+r_{-})^{2}-4m^{2}\right)+a^{2}(r_{+}-r_{-})^{2}}\rho^{2}\\
&-\frac{\left[am(r_{+}+r_{-}+2m)(4m_0^{2}-(r_{+}-r_{-})^{2})\right]^{2}}
{\left[m_0^{2}\left((r_{+}+r_{-})^{2}-4m^{2}\right)+a^{2}(r_{+}-r_{-})^{2}\right]
\left[m_0^{2}\left(r_{+}+r_{-}+2m\right)^{2}+a^{2}(r_{+}-r_{-})^{2}\right]},
\end{split}
\end{align}
with $r_{\pm}=\sqrt{\rho^{2}+(z\pm m_0)^{2}}$. Write $U_{kerr}=U_0+\overline{U}_{kerr}$ and $\alpha_{kerr}=\alpha_0+\overline{\alpha}_{kerr}$, where $U_0$ and $\alpha_0$ are the corresponding Schwarzschild functions. It follows that for $|z|<m_0$ we have
\begin{equation}
\overline{U}_{kerr}(0,z)=-\frac{1}{2}
\log\left(\frac{m^{2}(m+m_{0})^{2}}{m_0^{2}(m+m_{0})^{2}+a^{2}z^{2}}\right),
\end{equation}
and
\begin{equation}
\overline{\alpha}_{kerr}(0,z)=\frac{1}{2}\log\frac{\left[m_0^{2}(m+m_{0})^{2}+a^{2}z^{2}\right]^{2}}
{4m_0^{4}m^{2}(m+m_{0})^{2}}.
\end{equation}
Notice that $\mathcal{J}=0$ implies that $\overline{U}_{kerr}(0,z)=\overline{\alpha}_{kerr}(0,z)=0$ as expected, since half the horizon rod length is given by
\begin{equation}
m_0=\sqrt{m^{2}-a^{2}}
=\sqrt{m^2-\frac{\mathcal{J}^{2}}{m^{2}}}.
\end{equation}
We now have that on the horizon rod
\begin{equation}
\overline{\beta}_{kerr}(0,z)=\overline{\alpha}_{kerr}(0,z)-2\overline{U}_{kerr}(0,z)
=\log\frac{m(m+m_{0})}{2m_0^{2}}\geq 0.
\end{equation}
Consider now the surface gravity of the Kerr black hole
\begin{equation}
\kappa=\frac{\sqrt{m^4-\mathcal{J}^2}}{2\left(m^3+m\sqrt{m^4-\mathcal{J}^2}\right)}.
\end{equation}
Comparing the two formulas produces
\begin{equation}
\overline{\beta}_{kerr}(0,z)=-\log (4m_0\kappa).
\end{equation}


\begin{thebibliography}{99}

\bibitem{Anglada} P. Anglada, \textit{Penrose-like inequality with angular momentum for minimal surfaces}, Class. Quantum Grav., \textbf{35} 2018, 045018. arXiv:1708.04646

\bibitem{Anglada1} P. Anglada, \textit{Penrose-like inequality with angular momentum for general horizons}, preprint, 2018. arXiv:1810.11321

\bibitem{Bekenstein} J. Bekenstein, \textit{A universal upper bound on the entropy to energy ratio for bounded systems}, Phys. Rev. D, \textbf{23} (1981), 287.

\bibitem{Bray} H. Bray, \textit{Proof of the Riemannian Penrose inequality using the positive mass theorem}, J. Differential Geom., \textbf{59} (2001), 177-267.



\bibitem{Brill} D. Brill, \textit{On the positive definite mass of the Bondi-Weber-Wheeler time-symmetric gravitational waves}, Ann. Phys., \textbf{7} (1959), 466-483.






\bibitem{Chrusciel} P. Chru\'{s}ciel, \textit{Mass and angular-momentum inequalities for axi-symmetric initial data sets. I. Positivity of Mass}, Ann. Phys., \textbf{323} (2008), 2566-2590. arXiv:0710.3680


\bibitem{ChruscielCosta} P. Chru\'{s}ciel, and J. Costa, \textit{Mass, angular-momentum and charge inequalities for axisymmetric initial data}, Class. Quantum Grav., \textbf{26} (2009), no. 23, 235013. arXiv:0909.5625


\bibitem{ChruscielLiWeinstein} P. Chru\'{s}ciel, Y. Li, and G. Weinstein, \textit{Mass and angular-momentum inequalities for axi-symmetric initial data sets. II. Angular Momentum}, Ann. Phys., \textbf{323} (2008), 2591-2613. arXiv:0712.4064

\bibitem{ChruscielNguyen} P. Chru\'{s}ciel, and L. Nguyen, \textit{A lower bound for the mass of axisymmetric connected black hole data sets}, Class. Quantum Grav., \textbf{28} (2011), 125001. arXiv:1102.1175


\bibitem{Costa} J. Costa, \textit{Proof of a Dain inequality with charge}, J. Phys. A, \textbf{43} (2010), no. 28, 285202. arXiv:0912.0838

\bibitem{Dain0} S. Dain, \textit{Proof of the angular momentum-mass inequality for axisymmetric black hole}, J. Differential Geom., \textbf{79} (2008), 33-67. arXiv:gr-qc/0606105

\bibitem{Dain} S. Dain, \textit{Geometric inequalities for axially symmetric black holes}, Class. Quantum Grav., \textbf{29} (2012), 073001. arXiv:1111.3615

\bibitem{DainGabach} S. Dain, M. Gabach-Clement, \textit{Geometrical inequalities bounding angular momentum and charges in General Relativity}, Living Rev. Relativ., \textbf{21} (2018), no. 5.




\bibitem{DainKhuriWeinsteinYamada} S. Dain, M. Khuri, G. Weinstein, and S. Yamada,
\textit{Lower bounds for the area of black holes in terms of mass, charge, and angular momentum}, Phys. Rev. D, \textbf{88} (2013), 024048. arXiv:1306.4739


\bibitem{ClementJaramilloReiris} M. Gabach-Clement, J. Jaramillo, and M. Reiris, \textit{Proof of the area-angular momentum-charge inequality for axisymmetric black holes}, Class. Quantum Grav., \textbf{30} (2012), 065017. arXiv:1207.6761


\bibitem{GH} G. Gibbons, and G. Holzegel, \textit{The positive mass and isoperimetric inequalities for axisymmetric black holes in four and five dimensions}, Class. Quantum Grav., \textbf{23} (2006), 6459-6478. arXiv:gr-qc/0606116




\bibitem{HuiskenIlmanen} G. Huisken, and T. Ilmanen, \textit{The inverse mean curvature flow and the Riemannian Penrose inequality}, J. Differential Geom., \textbf{59} (2001), 353-437.



\bibitem{JaraczKhuri} J. Jaracz, and M. Khuri, \textit{Bekenstein bounds, Penrose inequalities, and black hole formation}, Phys. Rev. D, \textbf{97} (2018), 124026. arXiv:1802.04438




\bibitem{KWcharge} M. Khuri, and G. Weinstein \textit{The positive mass theorem for multiple rotating charged black holes}, Calc. Var. Partial Differential Equations, \textbf{55} (2016), no. 2, 1-29. arXiv.1502.06290v2

\bibitem{KhuriWeinsteinYamada} M. Khuri, G. Weinstein, and S. Yamada, \textit{Extensions of the charged Riemannian Penrose inequality}, Class. Quantum Grav., \textbf{32} (2015), 035019. arXiv:1410.5027

\bibitem{KhuriWeinsteinYamada1} M. Khuri, G. Weinstein, and S. Yamada, \textit{Proof of the Riemannian Penrose inequality with charge for multiple black holes}, J. Differential Geom., \textbf{106} (2017), 451-498. arXiv:1409.3271


\bibitem{Mars} M. Mars, \textit{Present status of the Penrose inequality}, Class. Quantum Grav., \textbf{26} (2009), no. 19, 193001. arXiv:0906.5566


\bibitem{Penrose} R. Penrose, \textit{Naked singularities}, Ann. New York Acad. Sci., \textbf{224} (1973), 125-134.


\bibitem{SchoenZhou} R. Schoen, and X. Zhou, \textit{Convexity of reduced energy and mass angular momentum inequalities}, Ann. Henri Poincar\'{e}, \textbf{14} (2013), 1747-1773. arXiv:1209.0019.


\bibitem{Stephani} H. Stephani, D. Kramer, M. Maccallum, C. Hoenselaers, and E. Herlt,
\textit{Exact Solutions to Einstein's Field Equations}, Cambridge University Press, 2nd Edition, 2003.

\bibitem{WeinsteinYamada} G. Weinstein, and S. Yamada, \textit{On a Penrose inequality with charge}, Comm. Math. Phys., \textbf{257} (2005), no. 3, 703-723. 	 arXiv:math/0405602


\end{thebibliography}
\end{document}